\documentclass[runningheads]{llncs}

\usepackage[T1]{fontenc}
\usepackage{graphicx}

\usepackage{amsbsy, enumerate, amsfonts, amsmath, mathabx, amssymb}
\usepackage{complexity}
\usepackage{tikz}
\usepackage[caption = false]{subfig}

\newcommand{\Set}[1]{\left\{\, #1 \,\right\}}
\newcommand{\bigO}[1]{\mathcal{O}(#1)}
\newcommand{\tildeO}[1]{\tilde{\mathcal{O}}(#1)}
\newcommand{\Span}[1]{\mathrm{span}(#1)}
\newcommand{\cost}[1]{\mathrm{cost}(#1)}

\newcommand{\order}[0]{\sigma}
\newcommand{\vl}[0]{\order_{l}}
\newcommand{\vr}[0]{\order_{r}}
\newcommand{\vi}[0]{\order_{i}}
\newcommand{\vj}[0]{\order_{j}}

\newcommand{\vv}[1]{\order_{#1}}

\newcommand{\ceil}[1]{\left\lceil #1 \right\rceil}

\newcommand{\opt}[1]{{#1}^{*}}
\newcommand{\reg}[1]{\bar{#1}}
\newcommand{\regular}[0]{\reg{\order}}
\newcommand{\rich}[0]{W_{I}}

\newcommand{\richin}[0]{W^{-}_{I}}
\newcommand{\richout}[0]{W^{+}_{I}}
\newcommand{\optorder}[0]{\opt{\order}}
\newcommand{\optI}[0]{I^{*}}
\newcommand{\optIL}[0]{L^{*}}
\newcommand{\optIR}[0]{R^{*}}
\newcommand{\IL}[0]{L}
\newcommand{\IR}[0]{R}

\newcommand{\appr}[0]{\alpha}
\newcommand{\maxlen}[0]{\beta}
\newcommand{\SFAST}[0]{\textsc{Subset-FAST}}
\newcommand{\SFAS}[0]{\textsc{Subset-FAS}}

\newcommand{\FAST}[0]{\textsc{FAST}}
\newcommand{\ins}[0]{\mathcal{I}}

\newcommand{\instance}[0]{\mathcal{I} = (D, T, k)}

\newcommand{\defproblem}[3]{
    \begin{center}
        \framebox{
            \parbox{0.92\textwidth}{
                #1 \\
                \textbf{Input}: #2 \\
                \textbf{Output}: #3
            }
        }
    \end{center}
}

\newtheorem{reduction}{Reduction Rule}
\newtheorem{corllary}{Corllary}

\begin{document}

\title{An Almost Quadratic Vertex Kernel for Subset Feedback Arc Set in Tournaments}

\titlerunning{Kernelizations for SFAST}

\author{anonymous}
\institute{anonymous}
\author{Tian Bai}

\authorrunning{T. Bai}

 \institute{
 The University of Hong Kong, Hong Kong, China \\
 University of Electronic Science and Technology of China, Chengdu, China \\
 \email{tianbai@hku.hk}
 }

\maketitle              

\begin{abstract}
    In the \textsc{Feedback Arc Set in Tournaments} (\FAST{}) problem, we are given a tournament $D$ and a positive integer $k$, and the objective is to determine whether there exists an arc set $S \subseteq A(D)$ of size at most $k$ whose removal makes the graph acyclic.
    This problem is well-known to be equivalent to a natural tournament ranking problem, whose task is to rank players in a tournament such that the number of pairs in which the lower-ranked player defeats the higher-ranked player is no more than $k$.
    Using the PTAS for \FAST{} [STOC 2007], Bessy et al.\ [JCSS 2011] present a $(2 + \varepsilon)k$-vertex kernel for this problem, given any fixed $\varepsilon > 0$.
    A generalization of \FAST{}, called \SFAST{}, further includes an additional terminal subset $T \subseteq V(D)$ in the input.
    The goal of \SFAST{} is to determine whether there is an arc set $S \subseteq A(D)$ of size at most $k$ whose removal ensures that no directed cycle passes through any terminal in $T$.
    Prior to our work, no polynomial kernel for \SFAST{} was known.
    In our work, we show that \SFAST{} admits an $\bigO{(\appr k)^{2}}$-vertex kernel, provided that \SFAST{} has an approximation algorithm with an approximation ratio $\appr$.
    Consequently, based on the known $\bigO{\log k \log \log k}$-approximation algorithm, we obtain an almost quadratic kernel for \SFAST{}.

\end{abstract}

\section{Introduction}

A tournament is a specialized form of a directed graph where every pair of vertices is connected by exactly one of the arcs, with two possible directions.
Given a tournament $D$ on $n$ vertices and a positive integer parameter $k$, the \textsc{Feedback Arc Set in Tornament} problem (\FAST{}) seeks to determine whether there exists a set of $k$ arcs whose removal results in an acyclic directed graph.

\FAST{} have several applications, including voting systems~\cite{4orCharonH07}, machine learning~\cite{jairCohenSS99}, search engine ranking~\cite{wwwDworkKNS01}, and rank aggregation~\cite{jacmAilonCN08,nipsVeerathuR21}.
Consider organizing a chess tournament where everyone plays against every other player, and the goal is to rank the players based on the results.
However, it is uncommon that the results are acyclic, making it impossible to rank the players using a straightforward topological ordering.
A natural objective is to find a ranking that minimizes the number of upsets, where an upset is defined as a pair of players in which the lower-ranked player defeats the higher-ranked player.
This problem is well-known to be equivalent to the \FAST{} problem~\cite{stocKenyon-MathieuS07,icalpAlonLS09}.

In some scenarios, it is necessary to rank only local players to determine advancements or similar purposes.
However, ranking based solely on the results among local players may overlook the advantages some local players have against global players.
Inspired by this observation, we investigate the problem of local ranking within a global context.
Given a set of local players, the task is to find a ranking that minimizes the number of upsets involving at least one local player -- either as one of the two players in the upset or as a player ranked between the two players in the upset.
This ranking problem is equivalent to a generalization of \FAST{}, known as \SFAST{}.
In \SFAST{}, we are additionally given a set of terminal vertices $T \subseteq V(D)$, and the objective is to find a set of arcs of size at most $k$ whose removal ensures that no directed cycle passes through any terminal in $T$.
When all vertices are terminals, i.e., $T = V(D)$, \SFAST{} degenerates to \FAST{}.

\FAST{} was shown to be $\NP$-complete by three independent works~\cite{siamdmAlon06,aaaiConitzer06,cpcCharbitTY07}, which implies that \SFAST{} is also $\NP$-complete.
From an approximation perspective, \FAST{} admits a polynomial time approximation scheme (PTAS)~\cite{stocKenyon-MathieuS07}, whereas no approximation algorithm with the constant approximation ratio is currently known for \SFAST{}.
The study of the parameterized complexity for \FAST{} has a rich history.
Raman and Saurabh~\cite{tcsRamanS06} showed that \FAST{} is fixed-parameter tractable by giving an algorithm with the running time $2.415^{k}n^{\bigO{1}}$.
Latter, Alon et al.\ \cite{icalpAlonLS09} gave a sub-exponential time parameterized algorithm running in time $2^{\bigO{\sqrt{k}\log^{2}k}}n^{\bigO{1}}$.
Currently, the fastest known algorithm for \FAST{} runs in $2^{\bigO{\sqrt{k}}}n^{\bigO{1}}$~\cite{isaacKarpinskiS10}.

A parameterized problem is said to admit a polynomial kernel if there exists a polynomial-time algorithm (with respect to the input size), known as a kernelization algorithm, that reduces the input instance to an instance whose size is bounded by a polynomial $p(k)$ in $k$, while preserving the answer.
 This reduced instance is referred to as a $p(k)$ kernel for the problem.
Kernelization has been a central focus of research in parameterized complexity in recent years.
By considering the triangles in tournaments, Dom et al.\ \cite{jdaDomGHNT10} and Alon et al.\ \cite{icalpAlonLS09} independently gave the $\bigO{k^{2}}$-vertex kernel for \FAST{}.
Finally, by using the PTAS for \FAST{}~\cite{stocKenyon-MathieuS07}, Bessy et al.\ \cite{jcssBessyFGPPST11} showed that for any fixed $\varepsilon > 0$, \FAST{} admits a $(2 + \varepsilon)k$-vertex kernel.
\SFAST{} is significantly more challenging.
To the best of our knowledge, no polynomial kernel was known before.

In this work, by introducing the novel concepts of regular orders and rich vertices, we present the first polynomial kernel for \SFAST{}.
We prove that \SFAST{} admits an $\bigO{(\appr k)^{2}}$-vertex kernel if \SFAST{} has an approximation algorithm with an approximation ratio $\appr$.
In general digraphs, it is known that \textsc{Subset Feedback Arc Set} (\SFAS) has an approximation algorithm with the approximation ratio $\appr = \bigO{\log k \log \log k}$~\cite{combinatoricaSeymour95,algorithmicaEvenNSS98}, we finally obtain an almost quadratic vertex kernel for \SFAST{}.

The paper is organized as follows.
In Section~\ref{SEC: PRE}, we give some definitions and preliminary results regarding \SFAST{}.
In Section~\ref{SEC: KERNEL}, we give an almost quadratic vertex kernel for $\SFAST$. 
Finally, we conclude with some remarks in Section~\ref{SEC: CONCLU}.

\section{Preliminary} \label{SEC: PRE}

\subsection{Digraphs and Tournaments}

In this paper, we work with a digraph $D$, where $V(D)$ denotes its vertex set and $A(D)$ denotes its arc set.
The number of vertices in $D$ is denoted by $n = |V(D)|$.
For a vertex $v \in V(D)$, we denote its \emph{in-neighborhood} by $N^{-}(v) \coloneq \{ u \in V(D) : uv \in A(D) \}$ and its \emph{out-neighborhood} by $N^{+}(v) \coloneq \{u \in V(D) : vu \in A(D) \}$.
Additionally, for any subset of vertices $X \subseteq V(D)$, we simply write $N^{-}_{X}(v) \coloneq N^{-}(v) \cap X$ and $N^{+}_{X}(v) \coloneq N^{+}(v) \cap X$.
For a vertex subset $X \subseteq V(D)$, the digraph induced by $X$ is denoted by $D[X]$, and we simply write $D - X \coloneq D[V(D) \setminus X]$.
For an arc subset $Y \subseteq A(D)$, we use $D - Y$ denote the digraph $D'$ obtained by removing arcs in $Y$ from $D$; that is, $V(D') = V(D)$ and $A(D') = A(D) \setminus Y$.

A \emph{tournament} is a directed graph with exactly one directed arc between every pair of vertices.
Let $\order = (\vv{1}, \vv{2}, \ldots, \vv{n})$ be an order of vertices, where $\vi$ is the $i$th vertex according to $\order$.
An \emph{interval} $I = [\vl, \vr]$ (w.r.t.\ $\order$) is defined as a vertex subset $\{ \vi \in V(D) : l \leq i \leq r \}$, and the length of the interval $I$ is expressed as $|I| = r - l + 1$.
For two vertices $\vi$ and $\vj$ with $1 \leq i < j \leq n$, we say $\vi$ is on the left of $\vj$, and $\vj$ is on the right of $\vi$.
A \emph{forward arc} (w.r.t.\ $\order$) is an arc $\vi\vj$ with $1 \leq i < j \leq n$, and a \emph{backward arc} (w.r.t.\ $\order$) is an arc $\vi\vj$ with $1 \leq j < i \leq n$.
Given a backward arc $e = \vr\vl$, the \emph{span} of $e$ is denoted by the interval $\Span{e} \coloneq [\vl, \vr]$.
For every vertex $v$ in the span of a backward $e$, we say that $e$ is above $v$.

\subsection{Subset Feedback Arc Set in Tournaments}

Given a tournament $D$ and a set of vertices $T \subseteq V(D)$ called terminal set, a cycle containing a terminal is called a \emph{$T$-cycle}.
A $T$-cycle of length three is called a \emph{$T$-triangle}.
An arc set $S$ is called a \emph{$T$-feedback arc set} if $D - S$ contains no $T$-cycles.
More formally, the problem we consider is defined as follows.

\defproblem
{\textsc{Subset Feedback Arc Set in Tournaments} (\SFAST)}
{A tournament $D$, a terminal set $T \subseteq V(D)$, and an integer $k \in \mathbb{N}$.}
{Determine whether there is a subset $S \subseteq A(D)$ of size at most $k$, whose removal makes the remaining graph without $T$-cycles.}

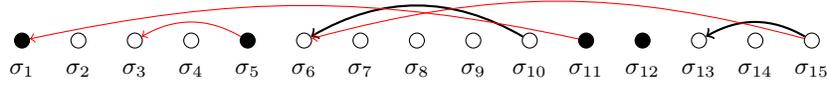
\begin{figure}[t!]
    \centering
        \begin{tikzpicture}
            [
            scale = 0.75,
            nonterminal/.style={draw, shape = circle, fill = white, inner sep = 2pt},
            terminal/.style={draw, shape = circle, fill = black, inner sep = 2pt},
            ]
            \node[terminal, label={[xshift = 0mm, yshift = -7mm]$\vv{1}$}] (v1) at(-7, 0) {};
            \node[nonterminal, label={[xshift = 0mm, yshift = -7mm]$\vv{2}$}] (v2) at(-6, 0) {};
            \node[nonterminal, label={[xshift = 0mm, yshift = -7mm]$\vv{3}$}] (v3) at(-5, 0) {};
            \node[nonterminal, label={[xshift = 0mm, yshift = -7mm]$\vv{4}$}] (v4) at(-4, 0) {};
            \node[terminal, label={[xshift = 0mm, yshift = -7mm]$\vv{5}$}] (v5) at(-3, 0) {};
            \node[nonterminal, label={[xshift = 0mm, yshift = -7mm]$\vv{6}$}] (v6) at(-2, 0) {};
            \node[nonterminal, label={[xshift = 0mm, yshift = -7mm]$\vv{7}$}] (v7) at(-1, 0) {};
            \node[nonterminal, label={[xshift = 0mm, yshift = -7mm]$\vv{8}$}] (v8) at(0, 0) {};
            \node[nonterminal, label={[xshift = 0mm, yshift = -7mm]$\vv{9}$}] (v9) at(1, 0) {};
            \node[nonterminal, label={[xshift = 0mm, yshift = -7mm]$\vv{10}$}] (v10) at(2, 0) {};
            \node[terminal, label={[xshift = 0mm, yshift = -7mm]$\vv{11}$}] (v11) at(3, 0) {};
            \node[terminal, label={[xshift = 0mm, yshift = -7mm]$\vv{12}$}] (v12) at(4, 0) {};
            \node[nonterminal, label={[xshift = 0mm, yshift = -7mm]$\vv{13}$}] (v13) at(5, 0) {};
            \node[nonterminal, label={[xshift = 0mm, yshift = -7mm]$\vv{14}$}] (v14) at(6, 0) {};
            \node[nonterminal, label={[xshift = 0mm, yshift = -7mm]$\vv{15}$}] (v15) at(7, 0) {};

            \draw[->, thick] (v10) to[out = 150, in = 30] (v6);
            \draw[->, thick] (v15) to[out = 150, in = 30] (v13);
            \draw[->, red] (v5) to[out = 150, in = 30] (v3);
            \draw[->, red] (v11) to[out = 168, in = 12] (v1);
            \draw[->, red] (v15) to[out = 165, in = 15] (v6);
        \end{tikzpicture}
    \caption{
        An order $\order$ of $15$ vertices with $\cost{\order} = 3$.
        Black vertices denote the terminals, and white vertices denote the non-terminals.
        In the order $\order$, black thick arcs present unaffected backward arcs, red arcs present affected arcs, and all forward arcs are omitted.
        The arc $\vv{5}\vv{3}$ is an affected arc above the terminal $\vv{5}$; the arc $\vv{11}\vv{1}$ is an affected arc above the terminals $\vv{1}$ and $\vv{11}$; the arc $\vv{15}\vv{6}$ is an affected arc above the terminals $\vv{11}$ and $\vv{12}$;
        The three intervals $[\vv{2}, \vv{4}]$, $[\vv{6}, \vv{10}]$, and $[\vv{13}, \vv{15}]$ are maximal non-terminal intervals w.r.t.\ the order $\order$.}
        \label{FIG: AFFECT}
\end{figure}

Given an order $\order$ of vertices, a \emph{non-temrinal interval} is an interval without terminals.
A non-terminal interval is \emph{maximal} if it is not the proper sub-interval of other non-terminal intervals.
Backward arcs w.r.t.\ $\order$ above some terminal are called \emph{affected arcs}.
A vertex is called \emph{affected} if it is the endpoint of an affected backward arc; otherwise, $v$ is called \emph{unaffected}.
The cost of an order $\order$, denoted by $\cost{\order}$, is defined as the number of affected arcs w.r.t.\ $\order$ (cf. Fig.~\ref{FIG: AFFECT}).

\begin{proposition} \label{PROP: RANK}
    A tournament $D$ has a $T$-feedback arc set of size at most $k$ if and only if there exits an order $\order$ of vertices with $\cost{\order} \leq k$.
\end{proposition}

\begin{proof}
    On the one hand, suppose there exists an arc set $S$ of size $|S| \leq k$ such that $D - S$ contains no $T$-cycle.
    In the subgraph $D - S$, every strongly connected component of size at least two contains no terminal.
    In addition, by considering the topological order of all strongly connected components, one can easily construct an order $\order$ of vertices such that the two endpoints of each backward belong to the same strongly connected component.
    This indicates that in the subgraph $D - S$, no backward arc is above any terminal.
    Consequently, in the tournament $D$, we have $\cost{\order} \leq |S| \leq k$.
    
    On the other hand, suppose there exists an order $\order$ of vertices with $\cost{\order} \leq k$.
    Let $S$ be the set of all affected arcs, then $|S| \leq k$.
    After removing $S$ from $D$, for any terminal $t \in T$, no backward arc is above $t$, which ensures that the terminal $t$ is not contained in any $T$-cycle in $D - S$.
    As a result, $S$ is a $T$-feedback arc set of $D$ with the size at most $k$.
    \qed
\end{proof}

Directly following Proposition~\ref{PROP: RANK}, we define the ranking version, which is equivalent to the original version of \SFAST{}.

\defproblem
{Ranking Version of \SFAST{}}
{A tournament $D$, a terminal set $T \subseteq V(D)$, and an integer $k \in \mathbb{N}$.}
{Determine whether there is an order $\order$ of vertices in $V(D)$ such that $\cost{\order} \leq k$.}

In our arguments, we will also need the following characterization.

\begin{proposition} \label{PROP: REVERSE}
    Suppose an arc $e$ is contained in a minimum $T$-feedback arc set of $D$.
    Let $D'$ be the tournament obtained from $D$ by reversing $e$.
    Then, $S$ is a minimum $T$-feedback arc set containing $e$ of $D$ if and only if $S \setminus \Set{e}$ is a minimum $T$-feedback arc set of $D'$.
\end{proposition}

\begin{proof}
    Suppose $S$ is a minimum $T$-feedback arc set of $D$.
    By Proposition~\ref{PROP: RANK}, there exists an order $\order$ of vertices with $\cost{\order} \leq |S|$.
    Observe that $S$ is of minimum size.
    It follows that $S$ must be exactly the set of all affected arc set w.r.t.\ $\order$.
    By construction, $D'$ is obtained from $D$ by reserving only $e$.
    Thus, $S \setminus \Set{e}$ is the set of all affected arcs w.r.t.\ $\order$ in $D'$, and by Proposition~\ref{PROP: RANK}, it is a minimum $T$-feedback arc set of $D'$.
    
    Conversely, suppose $S \setminus \Set{e}$ is a minimum $T$-feedback arc set of $D'$.
    It is clear that $S$ is a $T$-feedback arc set of $D$.
    Assume for contradiction that $S$ is not minimum.
    By the condition of the lemma, there exists a minimum $T$-feedback arc set $S'$ of $D$, such that $e \in S'$ and $|S'| < |S|$.
    It follows that $S' \setminus \Set{e}$ is a $T$-feedback arc set of $D'$ with the size smaller than $S \setminus \Set{e}$, contradicting the minimality of $S \setminus \Set{e}$.
    This completes our proof.
    \qed
\end{proof}

With the help of Proposition~\ref{PROP: REVERSE}, if we know that an arc belongs to a minimum subset feedback arc set, we can reduce the instance by reversing the arc and decreasing the $k$ by $1$. 

\section{Kernelization for \SFAST{}} \label{SEC: KERNEL}

In this section, we begin by introducing the concept of regular orders and proposing several basic reduction rules.
Subsequently, by analyzing the structural properties of optimal orders, we establish an almost quadratic kernel for \SFAST{}.

\subsection{Regular Orders} \label{SUBSEC: REGULAR}

This subsection focuses on the notion of \emph{regular} orders.
For any given vertex order $\order = (\vv{1}, \vv{2}, \ldots, \vv{n})$, we can construct a regular order $\regular = (\regular_{1}, \regular_{2}, \ldots, \regular_{n})$ in polynomial time that preserves the same cost while ensuring that the number of backward arcs of $\regular$ is no more than that of $\order$.
We refer to $\regular$ as a regularization of $\order$.

\begin{definition}[regular orders] \label{DEF: REGULAR}
    Given an instance $\instance$ of \SFAST{}, an order $\regular$ of vertices is called \emph{regular} if for any non-terminal interval $I = [\regular_{l}, \regular_{r}]$, the following conditions hold:
    \begin{itemize}
        \item $\regular_{l}$ has at least $\ceil{(r - l)/2}$ out-neighbors in $I$, i.e., $|N^{+}_{I}(\regular_{l})| \geq \ceil{(r - l)/2}$; and 
        \item $\regular_{r}$ has at least $\ceil{(r - l)/2}$ in-neighbors in $I$, i.e., $|N^{-}_{I}(\regular_{r})| \geq \ceil{(r - l)/2}$.
    \end{itemize}
\end{definition}

\begin{definition}[regularization] \label{DEF: REGULARIZATION}
    Given an instance $\instance$ of \SFAST{}, let $\regular$ be an order of vertices.
    A regular order $\regular$ of vertices is called a \emph{regularization} of $\order$ if it satisfies that
    \begin{itemize}
        \item two orders $\order$ and $\regular$ have the identical number of affected arcs, i.e., $\cost{\order} = \cost{\regular}$; and
        \item each terminal $t \in T$ has the identical index in $\order$ and $\regular$, i.e., $t = \vi = \regular_{i}$. 
    \end{itemize}
\end{definition}

\begin{lemma} \label{LEM: REGULAR}
    Let $\instance$ be an instance of \SFAST{}.
    For any order $\order$ of the vertices, there exists at least one regularization of $\order$.
    Furthermore, a regularization $\regular$ of $\order$ can be found in polynomial time.
\end{lemma}

\begin{proof}
    The following local search algorithm is used to obtain the desired regular order.
    At each step, the algorithm maintains the order $\regular$, starting with the initial order $\regular = \order$.
    As long as a non-terminal interval $I = [\regular_{l}, \regular_{r}]$ satisfies one of the following two cases, the algorithm performs the corresponding operation.
    \begin{enumerate}[{Case} 1:]
        \item If $|N^{+}_{I}(\regular_{l})| < |N^{-}_{I}(\regular_{l})|$, the algorithm modifies $\regular$ by moving $\regular_{l}$ to the right of $\regular_{r}$.
        \item If $|N^{-}_{I}(\regular_{r})| < |N^{+}_{I}(\regular_{r})|$, the algorithm modifies $\regular$ by moving $\regular_{r}$ to the left of $\regular_{l}$.
    \end{enumerate}
    The algorithm terminates when no non-terminal interval satisfies either of the above cases.
    
    When the algorithm terminates, for every non-terminal interval $I = [\regular_{l}, \regular_{r}]$, the order $\regular$ satisfies
    \begin{equation*}
        |N^{+}_{I}(\regular_{l})| 
        \geq |N^{-}_{I}(\regular_{l})| 
        = (|I| - 1) - |N^{+}_{I}(\regular_{l})| 
        = (r - l) - |N^{+}_{I}(\regular_{l})|,
    \end{equation*}
    and
    \begin{equation*}
        |N^{-}_{I}(\regular_{r})| 
        \geq |N^{+}_{I}(\regular_{r})| 
        = (|I| - 1) - |N^{-}_{I}(\regular_{r})|
        = (r - l) - |N^{-}_{I}(\regular_{r})|,
    \end{equation*}
    which implies that
    \begin{equation*}
        |N^{+}_{I}(\regular_{l})| \geq \ceil{(r - l)/2},
    \end{equation*}
    and
    \begin{equation*}
        |N^{-}_{I}(\regular_{r})| \geq \ceil{(r - l)/2}.
    \end{equation*}
    Thus, by Definition~\ref{DEF: REGULAR}, $\regular$ is a regular order when the algorithm terminates.
    
    Next, observe that each operation only adjusts the sub-order of vertices within the non-terminal interval $I$, ensuring that the index of every terminal remains unchanged.
    Additionally, the endpoints of any affected arc do not belong to the same non-terminal interval.
    This also implies that a backward arc is affected w.r.t.\ $\order$ if and only if it is affected w.r.t.\ $\regular$.
    As a result, $\bar{\order}$ is a regularization of $\order$.
    
    Finally, we show that the algorithm terminates in polynominal time.
    Note that there are at most $n^{2}$ distinct non-terminal intervals, and thus the algorithm can check whether any non-terminal interval satisfies one of the two cases in time $\bigO{n^{2}}$ per step.
    Furthermore, in each step, the number of backward arcs is decreased by at least $1$ (specifically, by $|N^{-}_{I}(\regular_{l})| - |N^{+}_{I}(\regular_{l})| \geq 1$ in case 1, and by $|N^{+}_{I}(\regular_{r})| - |N^{-}_{I}(\regular_{r})| \geq 1$ in case 2).
    Therefore, the algorithm terminates after performing at most $\bigO{n^{2}}$ steps.
    We conclude that a regularization of $\order$ can be found in time $\bigO{n^{4}}$, completing the proof.
    \qed
\end{proof}

Lemma~\ref{LEM: REGULAR} shows that a regularization of $\order$ can be found in polynomial time using the local search approach.
We note that the regularization of an order is not necessarily unique, as different sequences of local operations may lead to distinct regular orders.
With this result, we have established a foundation for efficiently transforming any given order into a regular one while preserving key properties, which will be instrumental in our subsequent analysis and kernelization design.

\subsection{Basic Reduction Rules} \label{SUBSEC: BASIC REDUCTION}

Our kernelization algorithm builds upon an approximation algorithm for \SFAST{}.
It is known that there exists an approximation algorithm for \SFAS{} with an approximation ratio of $\appr = \bigO{\log k \log \log k}$~\cite{combinatoricaSeymour95,algorithmicaEvenNSS98}.
This allows us to find an order $\order$ of vertices with $\cost{\order} \leq \appr k$; otherwise, we can immediately return a trivial NO-instance.
Furthermore, as discussed in Section~\ref{SUBSEC: REGULAR}, we can construct a regularization of any order of vertices while preserving its cost.
Therefore, throughout this subsection, we assume without loss of generality that the order is regular.

In our kernelization, we give several \emph{reduction rules} to obtain a kernel for \SFAST{}.
Each reduction rule transforms the input instance into another instance.
A reduction rule is said to be \emph{safe} if the input instance is a YES-instance if and only if the output instance is a YES-instance.
We note that some reduction rules rely on the regular order of vertices.
When applying a reduction rule, a vertex may be removed or an arc may be reversed, and the solution size parameter $k$ does not increase.
Since the order after applying a reduction rule may no longer be regular, the algorithm will compute a regularization of the new order, and the cost of the new order is also bounded by $\appr k$.

In this subsection, we provide four basic reduction rules.

\begin{reduction} \label{REDUCTION: NO}
    If $k \leq 0$ and the tournament $D$ contains a $T$-cycle, then return a trivial NO-instance.
\end{reduction}

\begin{reduction} \label{REDUCTION: YES}
    If $k \geq 0$ and the tournament $D$ contains no $T$-cycle, then return a trivial YES-instance.
\end{reduction}

The safeness of Reduction Rules~\ref{REDUCTION: NO} and~\ref{REDUCTION: YES} are straightforward.
Additionally, by leveraging the topological order, we can determine in polynomial time whether a tournament $D$ with the terminal set $T$ contains a $T$-cycle.

\begin{reduction} \label{REDUCTION: NONCYCLE}
    If a vertex $v$ is not contained in any $T$-cycle, then delete $v$ from $D$.
\end{reduction}

\begin{reduction} \label{REDUCTION: K CYCLES}
    Let $e = uv$ be an affected arc above a terminal $t \in T$.
    If there are $k + 1$ arc-disjoint paths from $v$ to $u$ that contain the terminal $t$ and consist solely of forward arcs, then reverse $e$ in $D$ and decrease $k$ by $1$.
\end{reduction}

It is easy to see that if Reduction Rule~\ref{REDUCTION: NONCYCLE} cannot be applied, for every terminal $t \in T$, there exists at least one affected backward arc above $t$.

\begin{lemma}
    Reduction Rules~\ref{REDUCTION: NONCYCLE} and~\ref{REDUCTION: K CYCLES} are safe and can be applied in polynomial time.
\end{lemma}

\begin{proof}
    We first consider Reduction Rule~\ref{REDUCTION: NONCYCLE}.
    Let $S$ be a $(T \setminus \Set{v})$-feedback arc set of $D - v$.
    By definition, we know $D - S - v$ contains no $T$-cycles.
    Since $v$ is not contained in any $T$-cycles, $D - S$ also contains no $T$-cycles.
    Thus, $S$ is a $T$-feedback arc set of the input graph $D$.

    Reduction Rule~\ref{REDUCTION: NONCYCLE} can be executed by checking whether there is a path from $N^{+}(v)$ to $N^{-}(v)$ in $D - v$ if $v$ is a terminal, or whether there is a path containing terminals from $N^{+}(v)$ to $N^{-}(v)$ in the $D - v$ if $v$ is a non-terminal.
    Additionally, after applying Reduction Rule~\ref{REDUCTION: NONCYCLE}, a vertex is removed, which means this rule can be executed at most $n$ times.
    Consequently, Reduction Rule~\ref{REDUCTION: NONCYCLE} can be applied in polynomial time.
     
    Next, we analyze Reduction Rule~\ref{REDUCTION: K CYCLES}.
    Observe that each path from $v$ to $u$ containing $t$, together with the backward arc $e$, forms a $T$-cycle.
    Thus, $e$ must belong to any $T$-feedback arc set, as any such solution set that excludes $e$ would contain at least one arc from each of the $k + 1$ paths, requiring a solution of size at least $k + 1$.
    Consequently, there exists an order with cost at most $k$ if and only if there exists an order with cost at most $k - 1$ after reversing $e$.
    This leads to the safeness of Reduction Rule~\ref{REDUCTION: K CYCLES}.

    Observe that forward arcs w.r.t.\ $\order$ form a directed acyclic graph $D'$.
    Let $f_{vt}$ (resp. $f_{tu}$) denote the maximum flow from $v$  to $t$ (resp. from $t$ to $u$) in $D'$, and both of which can be computed in polynomial time.
    Since the maximum number of arc-disjoint arcs from $v$ to $u$ containing $t$ and consisting of forward arcs equals $\min\Set{f_{ut}, f_{tv}}$, Reduction Rule~\ref{REDUCTION: K CYCLES} can be executed in polynomial time.
    Additionally, after applying Reduction Rule~\ref{REDUCTION: K CYCLES}, $k$ is decreased by one, meaning this rule can be executed at most $k$ times.   
    Consequently, Reduction Rule~\ref{REDUCTION: K CYCLES} can be also applied in polynomial time.
    \qed
\end{proof}

\begin{definition}[reduced instance]
    If Reduction Rules~\ref{REDUCTION: NO}-\ref{REDUCTION: K CYCLES} cannot be applied, the instance is called \emph{reduced}.
\end{definition}

Now, we analyze the properties of the reduced instance.
These properties characterize the possible positions of terminals within a regular order.

\begin{lemma} \label{LEM: TERMINAL LOCATION}
    Let $\instance$ be a reduced instance and $\order$ be a regular order of the vertices.
    Let $t$ be a terminal and $e$ be an affected arc above $t$.
    Suppose $\Span{e} = [\vl, \vr]$, we have $t \in [\vl, \vv{l + \ell}]$ or $t \in [\vv{r - \ell}, \vr]$, where $\ell = 2(\appr + 1) k + 2$.
\end{lemma}

\begin{proof}
    Suppose $t$ is the $i$th vertex according to the order $\order$, i.e., $t = \vi$.
    Let $\vj$ be the terminal in the interval $J = [\vl, \vi]$ with the minimum index.
    Notice that such a terminal $\vj$ always exists as $\vi$ is a terminal.
    We define two intervals $X = [\vl, \vv{j - 1}]$ and $Y = [\vj, \vi]$.
    Then, we know $l \leq j \leq i$ and $X$ is a non-terminal interval.
    
    Assume to the contrary that $t \in [\vv{l + \ell + 1}, \vv{r - \ell - 1}]$, i.e., $i \geq l + \ell + 1$.
       
    On the one hand, since $\order$ is regular, there are at least $\ceil{(j - 1 - l)/2}$ out-neighbors of $\vl$ in $X$.
    Thus, we have
    \begin{align*}
        |N^{+}_{X}(\vl) \cap N^{-}_{X}(\vi)| 
        \geq{}& |N^{+}_{X}(\vl)| - |N^{+}_{X}(\vi)| \\
        \geq{}& \ceil{\frac{j - 1 - l}{2}} - |N^{+}_{X}(\vi)| \\
        \geq{}& \frac{j - l - 1}{2} - |N^{+}_{X}(\vi)|.
    \end{align*}
    On the other hand, by the assumption that $i \geq l + \ell + 1$, we have
    \begin{align*}
        |N^{+}_{Y}(\vl) \cap N^{-}_{Y}(\vi)| 
        \geq{}& \max \Set{(|Y| - 1) -  |N^{-}_{Y}(\vl)| - |N^{+}_{Y}(\vi)|, 0} \\
        ={}& \max \Set{i - j - |N^{-}_{Y}(\vl)| - |N^{+}_{Y}(\vi)|, 0} \\
        ={}& \max \Set{\ell - (j - l - 1) - |N^{-}_{Y}(\vl)| - |N^{+}_{Y}(\vi)|, 0}.
    \end{align*}
    Furthermore, for any vertex $u \in N^{+}_{X}(\vi)$ or $u \in N^{+}_{Y}(\vi)$, the arc $tu$ is an affected backward arc above the terminal $t$; for any vertex $u \in N^{-}_{Y}(\vl)$, the arc $u\vl$ is an affected backward arc above the terminal $\vj$.
    Because the number of affected arcs is upper bounded by $\appr k$, we have
    \begin{equation*}
        |N^{+}_{X}(\vi)| + |N^{+}_{Y}(\vi)| + |N^{-}_{Y}(\vl)| \leq \appr k,
    \end{equation*}
    Combining the above relations, we get
    \begin{align*}
        |N^{+}_{J}(\vl) \cap N^{-}_{J}(\vi)|
        ={}& |N^{+}_{X}(\vl) \cap N^{-}_{X}(\vi)| + |N^{+}_{Y}(\vl) \cap N^{-}_{Y}(\vi)| \\
        \geq{}& \max\Set{\ell - (j - l - 1) - \appr k, \frac{j - l - 1}{2} - \appr k} \\
        \geq{}& \frac{\ell}{2} - \appr k \\
        ={}& k + 1.
    \end{align*}

    Therefore, there are at least $k + 1$ pairwise arc-disjoint paths from $\vl$ to $t$ formed by forward arcs.
    Similarly, we can also derive that there are at least $k + 1$ pairwise arc-disjoint paths from $t$ to $\vr$ formed by forward arcs.    
    Based on these two facts, there are $k + 1$ pairwise arc-disjoint paths from $\vl$ to $\vr$ containing $t$ formed by forward arcs.
    This indicates that Reduction Rule~\ref{REDUCTION: K CYCLES} can be applied, contradicting the assumption that $\ins$ is reduced.
    \qed
\end{proof}

\begin{lemma} \label{LEM: NONTERMINA INTERVAL LENGTH}
    Let $\instance$ be a reduced instance and $\order$ be a regular order of vertices.
    If the length of each non-terminal interval is upper bounded by $\maxlen = \maxlen(\appr, k)$ and $\ins$ contains more than
    \begin{equation*}
        (2 \appr k + 1) \maxlen + 4(\appr + 1)\appr k^{2} + 4 \appr k
    \end{equation*}
    vertices, then $\ins$ is a NO-instance.
\end{lemma}

\begin{proof}
    We consider an affected backward arc $e$ with its span $\Span{e} = [\vl, \vr]$.
    Let $\ell = 2(\appr + 1)k + 2$, and we divide $\Span{e}$ into three intervals $I_{1} = [\vl, \vl + \ell]$, $I_{2} = [\vl + \ell + 1, \vr - \ell - 1]$, and $I_{3} = [\vr - \ell, \vr]$.
    Based on the Lemma~\ref{LEM: TERMINAL LOCATION}, we know that $I_{2}$ is a non-terminal interval.
    It follows that
    \begin{align*}
        |\Span{e}| 
        ={}& |I_{1}| + |I_{2}| + |I_{3}| \\
        \leq{}& (\ell + 1) + \maxlen + (\ell + 1) \\
        ={}& \maxlen + 4(\appr + 1)k + 4.
    \end{align*}
    
    Next, we define $V_{1} \subseteq V$ as the set of vertices contained in the span of some affected arc, i.e.,
    \begin{equation*}
        V_{1} = \bigcup_{e\colon \Span{e} \cap T \neq \varnothing} \Span{e},
    \end{equation*}
    and $V_{2} = V \setminus V_{1}$.
    On the one hand, according to the construction of the order $\order$, there are at most $\appr k$ affected arcs.
    Thus, we have the upper bound of the size of $V_{1}$:
    \begin{equation*}
        |V_{1}| 
        \leq \appr k \cdot (\maxlen + 4(\appr + 1)k + 4).
    \end{equation*}
    On the other hand, $V_{2}$ is divided into at most $\appr k + 1$ intervals by the spans of all affected arcs.
    Additionally, since Reduction Rule~\ref{REDUCTION: NONCYCLE} cannot be applied, every terminal must be contained in a $T$-cycle, which indicates that no terminal belongs to $V_{2}$.
    Thus, $V_{2}$ is the union of at most $\appr k + 1$ non-terminal intervals, and we have
    \begin{equation*}
        |V_{2}| \leq (\appr k + 1) \maxlen.
    \end{equation*}

    Consequently, we finally obtain that
    \begin{align*}
        |V| 
        ={}& |V_{1}| + |V_{2}| \\
        \leq{}& \appr k \cdot (\maxlen + 4(\appr + 1)k + 4) + (\appr k + 1) \maxlen \\
        ={}& (2 \appr k + 1) \maxlen + 4(\appr + 1)\appr k^{2} + 4\appr k.
    \end{align*}
    This completes our proof.
    \qed
\end{proof}

According to Lemma~\ref{LEM: NONTERMINA INTERVAL LENGTH}, to obtain a polynomial kernel for \SFAST{}, it is sufficient to bound the length of the maximal non-terminal interval.
Recall that a non-terminal interval is maximal if it is not the proper sub-interval of other non-terminal intervals.
In subsection~\ref{SUBSEC: QUADRATIC KERNEL}, we present a kernel such that the length of the maximal non-terminal interval $\maxlen = \bigO{(\appr k)^{2}}$ w.r.t.\ a given regular order $\order$, which gives a kernel of $\bigO{(\appr k)^{2}}$ vertices.

\subsection{An Almost Quadratic Vertex Kernel} \label{SUBSEC: QUADRATIC KERNEL}

In this subsection, we introduce the notation of the \emph{rich vertices} with respect to a given order $\order$ with $\cost{\order} = \appr k$.
A vertex is called rich if it is a non-terminal with at least $d$ in-neighbors and $d$ out-neighbors in the same maximal non-terminal interval $I$, where $d = (\appr + 2)k + 1$.
By capturing the characterizations of the optimal orders, we show that when the length of a maximal non-terminal interval $I$ is long enough, we can safely delete a rich vertex from the interval $I$.

We begin by providing the formal definition of rich/in-rich/out-rich vertices.
Next, with the help of the regularization of orders, we prove that the number of vertices that are not rich in a maximal non-terminal interval is upper bounded by $4d$.
Finally, we introduce a key reduction rule to obtain an almost quadratic kernel for \SFAST{}.
The kernelization algorithm described in this part follows from Subsection~\ref{SUBSEC: BASIC REDUCTION}.
Therefore, we assume throughout that the input instance has already been reduced and that the vertex order is regular.

\begin{definition}[rich vertices]
    Let $\instance$ be the instance of \SFAST{} and $\order$ be a regular order of the vertices.
    Define $d = (\appr + 2)k + 1$.
    Let $v$ be a non-terminal and $I = [\vl, \vr]$ be the maximal non-terminal interval containing $v$.
    \begin{itemize}
        \item Vertex $v$ is called \emph{in-rich} \textup{(}w.r.t.\ $\order$\textup{)} if it has at most $d - 1$ out-neighbors in $I$, i.e., $|N^{+}_{I}(v)| \leq d - 1$;
        \item vertex $v$ is called \emph{out-rich} \textup{(}w.r.t.\ $\order$\textup{)} if it has at most $d - 1$ in-neighbors in $I$, i.e., $|N^{-}_{I}(v)| \leq d - 1$;
        \item vertex $v$ is called \emph{rich} \textup{(}w.r.t.\ $\order$\textup{)} if it has at least $d$ out-neighbors and at least $d$ in-neighbors in $I$, i.e.,$|N^{+}_{I}(v)| \geq d$ and $|N^{-}_{I}(v)| \geq d$.
    \end{itemize}
    Moreover, the set of in-rich vertices in $I$ is denoted by $\richin$; the set of out-rich vertices in $I$ is denoted by $\richout$; and the set of rich vertices in $I$ is denoted by $\rich$.
\end{definition}

Notice that by definition, the rich, in-rich, and out-rich vertices in the maximal non-terminal interval $I$ form a partition of $I$.

\begin{lemma} \label{LEM: RICH}
    Let $\instance$ be the instance of \SFAST{} and $\order$ be an order of the vertices.
    Define $d = (\appr + 2)k + 1$.
    The sizes of rich, in-rich, and out-rich vertices in a maximal non-terminal interval $I$ satisfy the following \textup{(}cf. Fig.~\ref{FIG: OPT}\textup{)}.
    \begin{itemize}
        \item at least $|I| - 4d$ vertices in $I$ are rich, i.e., $|\rich| \geq |I| - 4d$;
        \item at most $2d$ vertices in $I$ are in-rich, i.e., $|\richin| \leq 2d$;
        \item at most $2d$ vertices in $I$ are out-rich, i.e., $|\richout| \leq 2d$.
    \end{itemize}
\end{lemma}

\begin{proof}
    Suppose $I = [\vl, \vr]$.
    Without loss of generalization, we assume that $\order$ is regular.
    
    Consider the vertex $\vi \in I$, where $l \leq i \leq r$.
    On the one hand, in the interval $[\vl, \vi]$, the number of in-neighbors of the vertex $\vi$ is at least $\ceil{(i - l)/2}$.
    On the other hand, the number of in-neighbors of the vertex $\vi$ is at least $\ceil{(r - i)/2}$.
    As a result, we derive that
    \begin{equation*}
        |N^{-}_{I}(\vi)| \geq \ceil{\frac{i - l}{2}},~\quad\text{and}\quad~|N^{+}_{I}(v)| \geq \ceil{\frac{r - i}{2}}.
    \end{equation*}

    Thus, when $l \leq i \leq r - 2d$, we have
    \begin{equation*}
        |N^{+}_{I}(\vi)| 
        \geq \ceil{\frac{r - i}{2}} 
        \geq \ceil{\frac{r - (r - 2d)}{2}}
        = d,
    \end{equation*}
    and when $l + 2d \leq i \leq r$, we have
    \begin{equation*}
        |N^{-}_{I}(\vi)| 
        \geq \ceil{\frac{i - l}{2}} 
        \geq \ceil{\frac{(l + 2d) - l}{2}}
        = d.
    \end{equation*}
    
     One the one hand, each vertex $\vi$ with $l + 2d \leq i \leq r - 2d$ is rich, and thus we have
    \begin{equation*}
        |\rich| \geq (r - 2d) - (l + 2d) + 1 = (r - l + 1) - 4d = |I| - 4d.
    \end{equation*}
    On the other hand, all vertices $\vi$ with $l \leq i \leq r - 2d$ are not out-rich, and all vertices $\vi$ with $l + 2d \leq i \leq r$ are not in-rich.
    It follows that
    \begin{equation*}
        |\richin| \leq r - (r - 2d + 1) + 1 = 2d,
    \end{equation*}
    and
    \begin{equation*}
        |\richout| \leq (l + 2d - 1) - l + 1 = 2d.
    \end{equation*}
    This completes our proof.    
    \qed
\end{proof}

Now, we analyze the structure properties of the optimal orders when a maximal interval has a length of at least $(4\appr + 9)k + 1$.
For ease of presentation, we adopt the following notation throughout the rest of this subsection. 
Given an instance $\instance$ and a regular order $\order$ of vertices, let $I = [\vl, \vr]$ be a maximal non-terminal interval.
We call the tripe $(\IL, I, \IR)$ as an $I$-partition, where 
\begin{align*}
    \IL = [\vv{1}, \vv{l - 1}],
    \quad&&\quad 
    \IR = [\vv{r + 1}, \vv{n}].
\end{align*}
Notice that $\IL$ or $\IR$ may be empty.
Furthermore, since $I$ is maximal, any backward arcs whose endpoints belong to different intervals of the $I$-partition must be affected.
This observation plays a crucial role in our subsequent analysis.

\begin{lemma} \label{LEM: MAX NONTERMINAL INTERVAL}
    For each unaffected vertex $v \in I$, all vertices in $\IL$ are the in-neighbors of $v$ and all vertices in $\IR$ are the out-neighbor of $v$, i.e., $\IL \subseteq N^{-}(v)$ and $\IR \subseteq N^{+}(v)$.
\end{lemma}

\begin{proof}
    By definition, we have $\IL = [\vv{1}, \vv{l - 1}]$, and $\IR = [\vv{r + 1}, \vv{n}]$.
    If $\IL$ is non-empty, vertex $\vv{l - 1}$ is a terminal as $I$ is a maximal non-terminal interval.
    For any vertex $u \in L$, if there is an arc from $v$ to $u$, the arc $vu$ is an affected backward arc above the terminal $v_{l - 1}$, contradicting that $v$ is unaffected.
    As a result, $uv$ is an arc in the tournament $D$, which leads to $\IL \subseteq N^{-}(v)$.
    Similarly, we can also obtain that $\IR \subseteq N^{+}(v)$.
    \qed
\end{proof}

\begin{figure}[t!]
    \centering
        \begin{tikzpicture}
            [
            scale = 0.75,
            nonterminal/.style={draw, shape = circle, fill = white, inner sep = 2pt},
            terminal/.style={draw, shape = circle, fill = black, inner sep = 2pt},
            ]
            \node[nonterminal, label={[xshift = 0mm, yshift = -7mm]$\vv{1}$}] (v1) at(-7, 0) {};
            \node[label={[xshift = 0mm, yshift = -3mm]$\cdots$}] (v2) at(-6.25, 0) {};
            \node[terminal, label={[xshift = 0mm, yshift = -7.3mm]$\vv{l - 1}$}] (v3) at(-5.5, 0) {};
            \node[nonterminal, label={[xshift = 0mm, yshift = -7mm]$\vl$}] (v4) at(-4, 0) {};
            \node[label={[xshift = 0mm, yshift = -3mm]$\cdots$}] (v5) at(-3.25, 0) {};
            \node[nonterminal, label={[xshift = 0mm, yshift = -7mm]$\vv{l + 2d}$}] (v6) at(-2.5, 0) {};
            \node[label={[xshift = 0mm, yshift = -3mm]$\cdots$}] (v7) at(-1.675, 0) {};
            \node[nonterminal, label={[xshift = 0mm, yshift = -7mm]$\vv{l + 2d + 1}$}] (v8) at(-0.85, 0) {};
            \node[label={[xshift = 0mm, yshift = -3mm]$\cdots$}] (v9) at(0, 0) {};
            \node[nonterminal, label={[xshift = 0mm, yshift = -7mm]$\vv{r - 2d - 1}$}] (v10) at(0.85, 0) {};
            \node[label={[xshift = 0mm, yshift = -3mm]$\cdots$}] (v11) at(1.675, 0) {};
            \node[nonterminal, label={[xshift = 0mm, yshift = -7mm]$\vv{r - 2d}$}] (v12) at(2.5, 0) {};
            \node[label={[xshift = 0mm, yshift = -3mm]$\cdots$}] (v13) at(3.25, 0) {};
            \node[nonterminal, label={[xshift = 0mm, yshift = -7mm]$\vr$}] (v14) at(4, 0) {};
            \node[terminal, label={[xshift = 0mm, yshift = -7.3mm]$\vv{r + 1}$}] (v15) at(5.5, 0) {};
            \node[label={[xshift = 0mm, yshift = -3mm]$\cdots$}] (v16) at(6.25, 0) {};
            \node[nonterminal, label={[xshift = 0mm, yshift = -7mm]$\vv{n}$}] (v17) at(7, 0) {};

            \node[label={[xshift = 0mm, yshift = -3mm]$\IL = [\vv{1}, \vv{l - 1}]$}] (x1) at(-6, 1) {};
            \node[label={[xshift = 0mm, yshift = -3mm]$I = [\vl, \vr]$}] (x2) at(0, 1) {};
            \node[label={[xshift = 0mm, yshift = -3mm]$\IR = [\vv{r + 1}, \vv{n}]$}] (x3) at(6, 1) {};

            \node[label={[xshift = 0mm, yshift = -3mm]$\richout \subseteq [\vl, \vv{l + 2d}]$}] (x4) at(-3.2, -1.5) {};
            \node[label={[xshift = 0mm, yshift = -3mm]$\rich \subseteq I$}] (x5) at(0, -1.5) {};
            \node[label={[xshift = 0mm, yshift = -3mm]$\richin \subseteq [\vv{r - 2d}, \vr]$}] (x6) at(3.2, -1.5) {};
            
            \draw[rounded corners, dotted]
            (-7.4, -1) rectangle (-4.8, 0.7)
            (4.8, -1) rectangle (7.4, 0.7)
            (-4.6, -1) rectangle (4.6, 0.7);

            \draw[rounded corners, dashed, blue]
            (-4.5, -0.8) rectangle (-1.7, 0.5);
            \draw[rounded corners, dashed, purple]
            (1.7, -0.8) rectangle (4.5, 0.5);

            \node[nonterminal, label={[xshift = 0mm, yshift = -7.5mm]$\optorder_{1}$}] (u1) at(-7, -4) {};
            \node[label={[xshift = 0mm, yshift = -3mm]$\cdots$}] (u2) at(-6, -4) {};
            \node[nonterminal, label={[xshift = 0mm, yshift = -7.5mm]$\optorder_{l' - 2}$}] (u3) at(-5, -4) {};
            \node[terminal, label={[xshift = 0mm, yshift = -7.5mm]$\optorder_{l' - 1}$}] (u4) at(-4, -4) {};
            \node[nonterminal, label={[xshift = 0mm, yshift = -7.5mm]$\optorder_{l'}$}] (u5) at(-3, -4) {};
            \node[nonterminal, label={[xshift = 0mm, yshift = -7.5mm]$\optorder_{l' + 1}$}] (u6) at(-2, -4) {};
            \node[label={[xshift = 0mm, yshift = -3mm]$\cdots$}] (u7) at(-1, -4) {};
            \node[nonterminal, label={[xshift = 0mm, yshift = -7.5mm]$\optorder_{i}$}] (u8) at(0, -4) {};
            \node[label={[xshift = 0mm, yshift = -3mm]$\cdots$}] (u9) at(1, -4) {};
            \node[nonterminal, label={[xshift = 0mm, yshift = -7.5mm]$\optorder_{r' - 1}$}] (u10) at(2, -4) {};
            \node[nonterminal, label={[xshift = 0mm, yshift = -7.5mm]$\optorder_{r'}$}] (u11) at(3, -4) {};
            \node[terminal, label={[xshift = 0mm, yshift = -7.5mm]$\optorder_{r' + 1}$}] (u12) at(4, -4) {};
            \node[terminal, label={[xshift = 0mm, yshift = -7.5mm]$\optorder_{r' + 2}$}] (u13) at(5, -4) {};
            \node[label={[xshift = 0mm, yshift = -3mm]$\cdots$}] (u14) at(6, -4) {};
            \node[nonterminal, label={[xshift = 0mm, yshift = -7.5mm]$\optorder_{n}$}] (u15) at(7, -4) {};

            \draw[rounded corners, dotted]
            (-7.4, -5) rectangle (-3.55, -3.3)
            (3.45, -5) rectangle (7.4, -3.3)
            (-3.35, -5) rectangle (3.35, -3.3);

            \draw[rounded corners, dashed, blue]
            (-7.3, -4.75) rectangle (3.3, -3.4);
            \draw[rounded corners, dashed, purple]
            (-3.3, -4.85) rectangle (7.3, -3.5);

            \node[label={[xshift = 0mm, yshift = -3mm]$\optIL = [\optorder_{1}, \optorder_{l' - 1}]$}] (x1) at(-5.5, -3) {};
            \node[label={[xshift = 0mm, yshift = -3mm]$\optI = [\optorder_{l'}, \optorder_{r'}]$}] (x2) at(0, -3) {};
            \node[label={[xshift = 0mm, yshift = -3mm]$\optIR = [\optorder_{r' + 1}, \optorder_{n}]$}] (x3) at(5.5, -3) {};

            \node[label={[xshift = 0mm, yshift = -3mm]$\IL, W^{+}_{I} \subseteq \optIL \cup \optI$}] (x7) at(-4, -5.5) {};
            \node[label={[xshift = 0mm, yshift = -3mm]$W_{I} \subseteq \optI$}] (x8) at(0, -5.5) {};
            \node[label={[xshift = 0mm, yshift = -3mm]$\IR, W^{-}_{I} \subseteq \optIR \cup \optI$}] (x9) at(4, -5.5) {};
            
        \end{tikzpicture}
    \caption{
        A regular order $\order$ and an optimal order $\optorder$ of $n$ vertices.
        Black vertices denote the terminals, and white vertices denote the non-terminals.
        In the order $\order$, $I = [\vl, \vr]$ is a maximal non-terminal interval, and $(\IL, I, \IR)$ is the $I$-partition, where $\IL = [\vv{1}, \vv{l - 1}]$ and $\IR = [\vv{r + 1}, \vv{n}]$.
        In the order $\optorder$, the interval $\optI = [\optorder_{l'}, \optorder_{r'}]$ is a maximal non-terminal interval w.r.t\ $\optorder$, and $(\optIL, \optI, \optIR)$ is the $\optI$-partition, where $\IL = [\optorder_{1}, \optorder_{l' - 1}]$ and $\IR = [\optorder_{r' + 1}, \optorder_{n}]$.
        All out-rich vertices in $I$ are in the blue dotted boxes, and all in-rich vertices in $I$ are in the purple dotted boxes, where $d = (\appr + 2)k + 1$.}
        \label{FIG: OPT}
\end{figure}
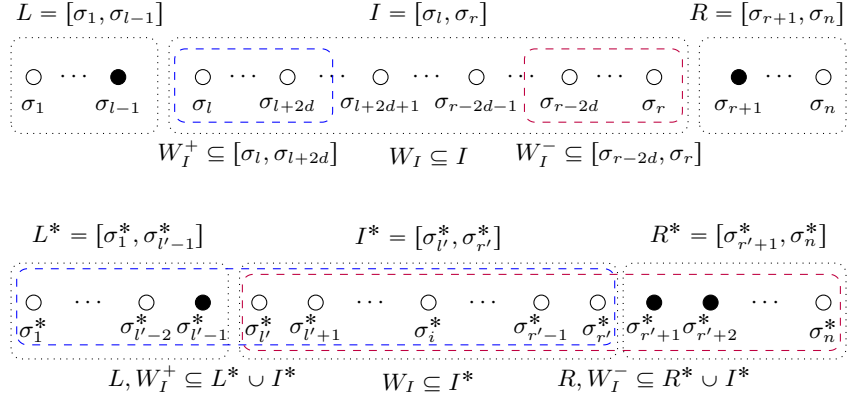

\begin{lemma} \label{LEM: TERMINAL}
    Let $d = (\appr + 2)k + 1$.
    If the input instance is a YES-instance and the interval $I$ has length at least $2d + 1$, in an optimal order $\optorder$ of vertices in $D$, we have
    \begin{itemize}
        \item rich or in-rich vertices in $I$ must lie on the left of any terminal in $\IL$; and
        \item rich or out-rich vertices in $I$ must lie on the right of any terminal in $\IR$.
    \end{itemize}
\end{lemma}

\begin{proof}    
    Let $t$ be a terminal in $\IL$ and $v$ be a rich or in-rich vertex in $I$.
    Notice that $tu$ is an arc for any unaffected vertex $u \in I$ by Lemma~\ref{LEM: MAX NONTERMINAL INTERVAL}.
    If $u$ is on the left of $t$ according to $\optorder$, then $tu$ is an affected arc w.r.t.\ $\optorder$.
    Thus, in the optimal order $\optorder$, at most $k$ unaffected vertices in $N^{-}_{I}(v)$ lie on the left of $t$.
    Since $\cost{\order} \leq \appr k$ and each affected arc has at most one endpoint in $I$, there are at most $\appr k$ affected vertices in $N^{-}_{I}(v)$.
    Additionally, because $|I| \geq 2d + 1$, we have 
    \begin{equation*}
        |N^{-}_{I}(v)| \geq d = (\appr + 2)k + 1.
    \end{equation*}    
    It follows that at least $k + 1$ unaffected in-neighbors of $v$ in $I$ are on the right of $t$ according to $\optorder$.    
    Consequently, in the optimal order $\optorder$, $v$ lies on the right of $t$; otherwise, there would be at least $k + 1$ affected arcs above $t$, contradicting that $\optorder$ is optimal.
    
    A symmetric argument shows that for any rich or out-rich vertex $v'$ in $I$ and any terminal $t'$ in $\IR$, $v'$ must lie on the right of $t'$ according to $\optorder$.
    This completes the proof.
    \qed
\end{proof}

\begin{lemma} \label{LEM: OPT}
    Let $d = (\appr + 2)k + 1$.
    If the input instance is a YES-instance, and the interval $I$ has length of at least $2d + 1$ and contains at least $k + 1$ rich vertices, then for any optimal order $\optorder$, there exists a maximal non-terminal interval $\optI$ with $\optI$-partition $(\optIL, \optI, \optIR)$ satisfying the following \textup{(}cf. Fig.~\ref{FIG: OPT}\textup{)}.
    \begin{enumerate}[\textup{(}a\textup{)}]
        \item all terminals in $\IL$ \textup{(}resp. $\IR$\textup{)} belong to $\optIL$ \textup{(}resp. $\optIR$\textup{)}, i.e., $T \cap \IL \subseteq \optIL$ and $T \cap \IR \subseteq \optIR$; \label{ITEM: TERMINAL}
        \item all rich vertices belong to $\optI$, i.e. $\rich \subseteq \optI$; and \label{ITEM: RICH}
        \item no vertices in $\richout$ \textup{(}resp. $\richin$\textup{)} belong to $\optIR$ \textup{(}resp. $\optIL$\textup{)}, i.e., $\richout \cap \optIR = \varnothing$ and $\richin \cap \optIL = \varnothing$. \label{ITEM: RICH IN+OUT}
        \item no vertices in $\IL$ \textup{(}resp. $\IR$\textup{)} belong to $\optIR$ \textup{(}resp. $\optIL$\textup{)}, i.e., $\IL \cap \optIR = \varnothing$ and $\IR \cap \optIL = \varnothing$. \label{ITEM: NONTERMINAL}
    \end{enumerate}
\end{lemma}

\begin{proof}
    Let $v \in \rich$ be any rich vertex w.r.t.\ $\order$ in $I$.
    Given the optimal order $\optorder$, denote by $\optI$ the maximal non-terminal interval (w.r.t.\ $\optorder$) containing $v$.
    We show that the $\optI$-partition $(\optIL, \optI, \optIR)$ satisfies the four properties in the lemma.
    
    We first consider (\ref{ITEM: TERMINAL}).
    According to Lemma~\ref{LEM: TERMINAL}, any terminal $t \in \IL$ lies on the left of the rich vertex $v$ according to $\optorder$.
    Since $\optI$ is a non-terminal interval, it follows that $t \in \optIL$.
    This implies the correctness of (\ref{ITEM: TERMINAL}).
    
    Then, we consider (\ref{ITEM: RICH}).
    By (\ref{ITEM: TERMINAL}), all rich vertices in $I$ must lie on the right of the terminals in $\optIL$ and on the left of the terminals in $\optIR$.
    Consequently, since $\optI$ is a maximal non-terminal interval, all rich vertices must belong to $\optI$, leading to the correctness of (\ref{ITEM: RICH}).

    Next, we consider (\ref{ITEM: RICH IN+OUT}).
    According to Lemma~\ref{LEM: TERMINAL}, in the optimal order $\optorder$, any out-rich vertex in $I$ lies on the right of each terminal in $\IL$.
    Since $\optI$ is a maximal non-terminal interval, no out-rich vertex in $I$ appears in $\optIL$, i.e., $\richin \cap \optIL = \varnothing$.
    Similarly, we can also obtain that $\richout \cap \optIR = \varnothing$.
    
    Finally, we consider (\ref{ITEM: NONTERMINAL}).
    By the correctness of (\ref{ITEM: TERMINAL}), we only need to focus on the non-terminals in $\IL$ and $\IR$.
    Let $u$ be a non-terminal in $\IL$.
    Assume for contradiction that $u \in \optIR$.
    For each vertex $v$ in $\optI$, Lemma~\ref{LEM: MAX NONTERMINAL INTERVAL} implies that $uv$ is an affected backward arc.
    Since $|\optI| \geq |\rich| \geq k + 1$, we have $\cost{\optorder} \geq k + 1$.
    This contradicts the optimality of $\optorder$.
    Similarly, we can show that the no non-terminals in $\IR$ belong to $\optIL$.
    Combining with (\ref{ITEM: TERMINAL}), we finally conclude that $\IL \cap \optIR = \varnothing$ and $\IR \cap \optIL = \varnothing$.
    \qed
\end{proof}

With the necessary groundwork established, we now proceed to introduce the key reduction rule.

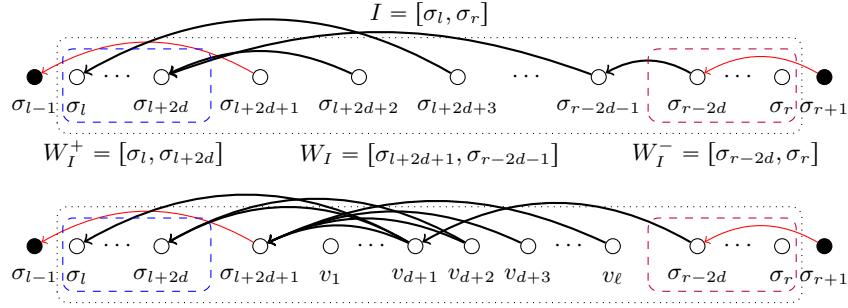
\begin{figure}[t!]
    \centering
        \begin{tikzpicture}
            [
            scale = 0.75,
            nonterminal/.style={draw, shape = circle, fill = white, inner sep = 2pt},
            terminal/.style={draw, shape = circle, fill = black, inner sep = 2pt},
            ]            
            \node[terminal, label={[xshift = 0mm, yshift = -7.3mm]$\vv{l - 1}$}] (v1) at(-7, 0) {};
            \node[nonterminal, label={[xshift = 0mm, yshift = -7.2mm]$\vl$}] (v2) at(-6.25, 0) {};
            \node[label={[xshift = 0mm, yshift = -3mm]$\cdots$}] (v3) at(-5.5, 0) {};
            \node[nonterminal, label={[xshift = 0mm, yshift = -7.2mm]$\vv{l + 2d}$}] (v4) at(-4.75, 0) {};
            \node[nonterminal, label={[xshift = 0mm, yshift = -7.3mm]$\vv{l + 2d + 1}$}] (v5) at(-3, 0) {};
            \node[nonterminal, label={[xshift = 0mm, yshift = -7.3mm]$\vv{l + 2d + 2}$}] (v6) at(-1.25, 0) {};
            \node[nonterminal, label={[xshift = 0mm, yshift = -7.3mm]$\vv{l + 2d + 3}$}] (v8) at(0.5, 0) {};
            \node[label={[xshift = 0mm, yshift = -3mm]$\cdots$}] (v9) at(1.75, 0) {};
            \node[nonterminal, label={[xshift = 0mm, yshift = -7.3mm]$\vv{r - 2d - 1}$}] (v12) at(3, 0) {};
            \node[nonterminal, label={[xshift = 0mm, yshift = -7.3mm]$\vv{r - 2d}$}] (v14) at(4.75, 0) {};
            \node[label={[xshift = 0mm, yshift = -3mm]$\cdots$}] (v15) at(5.5, 0) {};
            \node[nonterminal, label={[xshift = 0mm, yshift = -7.2mm]$\vr$}] (v16) at(6.25, 0) {};
            \node[terminal, label={[xshift = 0mm, yshift = -7.5mm]$\vv{r + 1}$}] (v17) at(7, 0) {};
            
            \node[label={[xshift = 0mm, yshift = -3mm]$I = [\vl, \vr]$}] (x2) at(0, 1) {};

            \node[label={[xshift = 0mm, yshift = -3mm]$\richout = [\vl, \vv{l + 2d}]$}] (x4) at(-5.25, -1.5) {};
            \node[label={[xshift = 0mm, yshift = -3mm]$\rich = [\vv{l + 2d + 1}, \vv{r - 2d - 1}]$}] (x5) at(0, -1.5) {};
            \node[label={[xshift = 0mm, yshift = -3mm]$\richin = [\vv{r - 2d}, \vr]$}] (x6) at(5.25, -1.5) {};

            \draw[->, red] (v5) to[out = 150, in = 30] (v1);
            \draw[->, red] (v17) to[out = 150, in = 30] (v14);
            \draw[->, thick] (v6) to[out = 155, in = 25] (v4);
            \draw[->, thick] (v12) to[out = 160, in = 20] (v4);
            \draw[->, thick] (v8) to[out = 150, in = 30] (v2);
            \draw[->, thick] (v14) to[out = 150, in = 30] (v12);
            
            \draw[rounded corners, dotted]
            (-6.6, -1) rectangle (6.6, 0.7);

             \draw[rounded corners, dashed, blue]
             (-6.5, -0.8) rectangle (-3.875, 0.5);
             \draw[rounded corners, dashed, purple]
             (3.875, -0.8) rectangle (6.5, 0.5);

            \node[terminal, label={[xshift = 0mm, yshift = -7.3mm]$\vv{l - 1}$}] (u1) at(-7, -3) {};
            \node[nonterminal, label={[xshift = 0mm, yshift = -7.2mm]$\vl$}] (u2) at(-6.25, -3) {};
            \node[label={[xshift = 0mm, yshift = -3mm]$\cdots$}] (u3) at(-5.5, -3) {};
            \node[nonterminal, label={[xshift = 0mm, yshift = -7.3mm]$\vv{l + 2d}$}] (u4) at(-4.75, -3) {};
            \node[nonterminal, label={[xshift = 0mm, yshift = -7.3mm]$\vv{l + 2d + 1}$}] (u5) at(-3, -3) {};
            \node[nonterminal, label={[xshift = 0mm, yshift = -7.2mm]$v_{1}$}] (u6) at(-1.75, -3) {};
            \node[label={[xshift = 0mm, yshift = -3mm]$\cdots$}] (u7) at(-1, -3) {};
            \node[nonterminal, label={[xshift = 0mm, yshift = -7.2mm]$v_{d + 1}$}] (u8) at(-0.25, -3) {};
            \node[nonterminal, label={[xshift = 0mm, yshift = -7.2mm]$v_{d + 2}$}] (u9) at(0.75, -3) {};
            \node[nonterminal, label={[xshift = 0mm, yshift = -7.2mm]$v_{d + 3}$}] (u10) at(1.75, -3) {};
            \node[label={[xshift = 0mm, yshift = -3mm]$\cdots$}] (u11) at(2.5, -3) {};
            \node[nonterminal, label={[xshift = 0mm, yshift = -7.3mm]$v_{\ell}$}] (u12) at(3.25, -3) {};
            \node[nonterminal, label={[xshift = 0mm, yshift = -7.3mm]$\vv{r - 2d}$}] (u14) at(4.75, -3) {};
            \node[label={[xshift = 0mm, yshift = -3mm]$\cdots$}] (v15) at(5.5, -3) {};
            \node[nonterminal, label={[xshift = 0mm, yshift = -7.2mm]$\vr$}] (u16) at(6.25, -3) {};
            \node[terminal, label={[xshift = 0mm, yshift = -7.5mm]$\vv{r + 1}$}] (u17) at(7, -3) {};

            \draw[->, red] (u5) to[out = 150, in = 30] (u1);
            \draw[->, red] (u17) to[out = 150, in = 30] (u14);
            \draw[->, thick] (u8) to[out = 150, in = 30] (u4);
            \draw[->, thick] (u8) to[out = 150, in = 30] (u2);
            \draw[->, thick] (u9) to[out = 150, in = 30] (u4);
            \draw[->, thick] (u14) to[out = 150, in = 30] (u8);
            \draw[->, thick] (u8) to[out = 155, in = 25] (u5);
            \draw[->, thick] (u9) to[out = 155, in = 25] (u5);
            \draw[->, thick] (u10) to[out = 155, in = 25] (u5);
            \draw[->, thick] (u12) to[out = 155, in = 25] (u5);

            \draw[rounded corners, dotted]
            (-6.6, -4) rectangle (6.6, -2.3);

            \draw[rounded corners, dashed, blue]
             (-6.5, -3.8) rectangle (-3.875, -2.5);
             \draw[rounded corners, dashed, purple]
             (3.875, -3.8) rectangle (6.5, -2.5);
            
        \end{tikzpicture}
    \caption{
        A maximal non-terminal interval $I$ w.r.t.\ an order $\order$ and the non-terminal interval obtained from $I$ by applying Reduction Rule~\ref{REDUCTION: RICH REPLACE}.
        Black vertices denote the terminals, and white vertices denote the non-terminals.
        In the order $\order$, $I = [\vl, \vr]$ is a maximal non-terminal interval, black thick arcs present unaffected backward arcs, red arcs present affected arcs, and all forward arcs are omitted.
        In the interval $I$, assume that $[\vl, \vv{l + 2d}]$ is the set of out-rich vertices (in the blue box), $\rich = [\vv{l + 2d + 1}, \vv{r - 2d - 1}]$ is the set of rich vertices, and $\richin = [\vv{r - 2d}, \vr]$ is the set of in-rich vertices (in the purple box).
        In the order $\order$, $\vv{l + 2d + 1}$ is an affected rich vertex and $\vv{r - 2d}$ is an affected in-rich vertex.
        After applying Reduction Rule~\ref{REDUCTION: RICH REPLACE}, the unaffected rich vertices $\vv{l + 2d + 2}, \vv{l + 2d +3}, \ldots, \vv{r - 2d - 1}$ are removed and $\ell$ non-terminals $v_{1}, v_{2}, \ldots, v_{\ell}$ are added.
        }
        \label{FIG: REPLACE}
\end{figure}

\begin{reduction} \label{REDUCTION: RICH REPLACE}
    Let $d = (\appr + 2)k + 1$ and $\ell = 2d + k + 1$.
    If there is a maximal non-terminal interval $I = [\vl, \vr]$ with length at least $(7\appr + 13)k + 6$, then we do the following seven steps \textup{(}cf.~\ref{FIG: REPLACE} \textup{)}.
    \begin{enumerate}
        \item Remove all unaffected rich vertices in $I$;
        \item add $\ell$ new non-terminals $v_{1}, v_{2}, \ldots, v_{\ell}$;
        \item for each pair of vertices $v_{i}$ and $v_{j}$ with $1 \leq i < j \leq \ell$, add arc $v_{i}v_{j}$;
        \item for each out-rich vertex $u \in \richout$ with $x$ unaffected rich in-neighbors in $I$, add arcs $uv_{1}, uv_{2}, \ldots, uv_{\ell}$ and reserve the arcs $uv_{d + 1}, uv_{d + 2}, \ldots, uv_{d + x}$;
        \item for each in-rich vertex $w \in \richout$ with $y$ unaffected rich in-neighbors in $I$, add arcs $v_{1}w, v_{2}w, \ldots, v_{\ell}w$ and reverse the arcs $v_{d + 1}w, v_{d + 2}w, \ldots, v_{d + y}w$;
        \item for each affected rich vertex $v' \in \rich$, add arcs $v_{1}v', v_{2}v', \ldots, v_{d}v'$ and arcs $v_{d + 1}v', v_{d + 2}v', \ldots, v_{\ell}v'$; and
        \item for each vertices $u' \in \IL$ and $w' \in \IR$, add arcs $u' v_{i}$ and $v_{i} w'$ for all $1 \leq i \leq \ell$.
    \end{enumerate}
\end{reduction}

Clearly, the digraph after applying Reduction Rule~\ref{REDUCTION: RICH REPLACE} is still a tournament.

\begin{lemma} \label{LEM: LINEAR LENTH}
    Let $d = (\appr + 2)k + 1$.
    If Reduction Rule~\ref{REDUCTION: RICH REPLACE} cannot be applied, we have the upper bound of the length of each non-terminal interval:
    \begin{equation*}
       \maxlen = (7\appr + 13)k + 5.
    \end{equation*}
\end{lemma}

\begin{proof}
    Denote by $U$ the set of all unaffected vertices in $I$ and denote $U' = \Set{v_{1}, v_{2}, \ldots, v_{\ell}}$.
    Let $I' \coloneq (I \setminus U) \cup U'$ be the non-terminal set obtained by removing all unaffected rich vertices in $I$ and adding the $\ell$ new non-terminals according to Reduction Rule~\ref{REDUCTION: RICH REPLACE}.
    
    According to Lemma~\ref{LEM: RICH}, there are at most $2d$ in-rich and $2d$ out-rich vertices in $I$.
    Additionally, since each affected arc has at most one endpoint in $I$, there are at most $\appr k$ affected vertices in $I$.
    After applying Reduction Rule~\ref{REDUCTION: RICH REPLACE}, $\ell = 2d + k + 1$ vertices are added, and all unaffected rich vertices are removed.
    Thus, we derive that
    \begin{align*}
        |I'| 
        \leq{}& 2d + 2d + \appr k + \ell \\
        ={}& 2((\appr + 2)k + 1) + 2((\appr + 2)k + 1) + \appr k + (2d + k + 1) \\
        ={}& (7\appr + 13)k + 5.
    \end{align*}
    Finally, we conclude that if Reduction Rule~\ref{REDUCTION: RICH REPLACE} cannot be applied, we have the length of each non-terminal interval is bound by $(7\appr + 13)k + 5$.
    \qed
\end{proof}

\begin{lemma} \label{SAFE: RICH REPLACE}
    Reduction Rule~\ref{REDUCTION: RICH REPLACE} is safe and can be applied in polynomial time.
\end{lemma}

\begin{proof}
    Denote by $U$ the set of all unaffected vertices in $I$ and denote $U' = \Set{v_{1}, v_{2}, \ldots, v_{\ell}}$.
    Let $D'$ be the tournament obtained from $D$ by applying Reduction Rule~\ref{REDUCTION: RICH REPLACE} and $n' = |V(D')|$.
    We show $\ins$ is a YES-instance if and only if $\ins'$ is a YES-instance.

    We first show that $\ins'$ is a YES-instance if $\ins$ is a YES-instance.
    
    Recall that $I = [\vl, \vr]$ is the maximal non-terminal interval w.r.t.\ $\order$.
    Define
    \begin{equation*}
        I' \coloneq (I \setminus U) \cup U'
    \end{equation*}
    as the non-terminal set obtained by removing all unaffected rich vertices in $I$ and adding the $\ell$ new non-terminals according to Reduction Rule~\ref{REDUCTION: RICH REPLACE}.
    Fix an order $\tau' = (\tau'_{1}, \tau'_{2}, \ldots, \tau'_{|I'|})$ of vertices in $I'$, and define
    \begin{equation*}
        \order' \coloneq (\vv{1}, \vv{2}, \ldots, \vv{l - 1}, \tau'_{1}, \tau'_{2}, \ldots, \tau'_{|I'|}, \vv{r + 1}, \vv{r + 2}, \ldots, \vv{n})
    \end{equation*}
    as an order of vertices in $V(D')$.
    Clearly, $I'$ is the maximal non-terminal interval w.r.t. $\order'$ since $\vv{l - 1}$ and $\vv{r + 1}$ are terminals.

    We now prove $\cost{\order} \leq \cost{\order'}$.
    Because of Step~7, for any vertex $v_{i} \in U'$, we have $\IL \subseteq N^{-}(v_{i})$ and $\IR \subseteq N^{+}(v_{i})$.
    This implies that all newly added vertices are unaffected w.r.t.\ $\order'$.
    Consequently, we conclude that $\cost{\order} \leq \cost{\order'}$, implying that $\ins'$ is a YES-instance if $\ins$ is a YES-instance.

    Next, we show $\ins$ is a YES-instance if $\ins'$ is a YES-instance.
    
    Suppose $\optorder$ is an optimal order of vertices in $V(D')$.
    Because of Step~3, every vertex $v_{i}$ with $d + 1 \leq i \leq d + k + 1$ has at least $d$ in-neighbors and $d$ out-neighbors in $I'$, and thus is rich in $I'$.
    It follows that $|W_{I'}| \geq k + 1$ holds.
    Additionally, since there are $\ell = 2d + k + 1$ non-terminals added, we know $|I'| \geq \ell > 2d + 1$.
    Hence, the conditions in Lemma~\ref{LEM: OPT} are satisfied.
    As a result, there exists a maximal non-terminal interval $\optI$ w.r.t.\ $\optorder$ with $\optI$-partition $(\optIL, \optI, \optIR)$ satisfying properties (\ref{ITEM: TERMINAL})-(\ref{ITEM: NONTERMINAL}).
    Let $(\IL', I', \IR')$ be the $I'$-partition.
    Based on Lemma~\ref{LEM: OPT}(\ref{ITEM: TERMINAL}) and~(\ref{ITEM: NONTERMINAL}), we have
    \begin{equation*}
        \IL = \IL' \subseteq \optIL,
        \quad\quad\quad
        \IR = \IR' \subseteq \optIR.
    \end{equation*}
    Because of Step~4, for any out-rich vertex $u \in \richout$, we have $|N^{-}_{I'}(u)| = |N^{-}_{I}(u)|$, leading that $u$ is also out-rich in $I'$.
    Similarly, any in-rich vertex $w \in \richin$ is also in-rich in $I'$.
    By Lemma~\ref{LEM: OPT}(\ref{ITEM: RICH IN+OUT}), it follows that
    \begin{equation*}
        \richout \subseteq W^{+}_{I'} \subseteq \optIL \cup \optI,
        \quad\quad\quad
        \richin \subseteq W^{-}_{I'} \subseteq \optIR \cup \optI.
    \end{equation*}
    Then, any vertex $v_{i} \in U'$ with $d + 1 \leq i \leq d + k + 1$ is rich.
    Besides, for any affected rich vertex $v' \in \rich \setminus U$, $v_{1}, v_{2}, \ldots, v_{d}$ are in-neighbors of $v'$ in $I'$ and $v_{d + 1}, v_{d + 2}, \ldots, v_{\ell}$ are out-neighbors of $v'$ in $I'$, because of Step~6.
    This indicates that $v'$ is also rich in $I'$.
    By Lemma~\ref{LEM: OPT}(\ref{ITEM: RICH}), we have
    \begin{equation*}
        v_{d + 1}, v_{d + 2}, \ldots, v_{d + k + 1} \in W_{I'} \subseteq \optI,
        \quad\quad\quad
        \rich \setminus U \subseteq W_{I'} \subseteq \optI.
    \end{equation*}
    Furthermore, we define
    \begin{equation*}
        I^{\dag} \coloneq (\optI \cup U) \setminus U'
    \end{equation*}
    as a set of non-terminals in $V(D)$.
    Fix an order $\tau = (\tau_{1}, \tau_{2}, \ldots, \tau_{|I^{\dag}|})$ of vertices in $I^{\dag}$, and define
    \begin{equation*}
        \order^{\dag} \coloneq (\optorder_{1}, \optorder_{2}, \ldots, \optorder_{l - 1}, \tau_{l}, \tau_{2}, \ldots, \tau_{|I^{\dag}|}, \optorder_{r + 1}, \optorder_{r + 2}, \ldots, \optorder_{n'}).
    \end{equation*}
    as an order $\order^{\dag}$ of vertices in $V(D)$.
    Clearly, $I^{\dag}$ is the maximal non-terminal interval w.r.t. $\order^{\dag}$ since $\optorder_{l - 1}$ and $\optorder_{r + 1}$ are terminals.
    Let $(\IL^{\dag}, I^{\dag}, \IR^{\dag})$ be the $I^{\dag}$-partition.
    By the construction of $\order^{\dag}$, we have $\IL^{\dag} = \optIL$ and $\IR^{\dag} = \optIR$.

    We now prove that $\cost{\order^{\dag}} \leq \cost{\optorder}$.
    Notice that in the order $\order^{\dag}$, all vertices in $U$ belong to the same non-terminal interval.
    This means that there are no affected arcs w.r.t.\ $\order^{\dag}$ that have both endpoints in $U$.
    Thus, we only need to consider affected arcs w.r.t.\ $\order^{\dag}$ that have one endpoint in $U$ and the other endpoint $v \in V(D) \setminus U$.
    We analyze these arcs by considering five cases according to the position of $v$ in $\order$.

    \textbf{Case 1:} $v \in \IR$.
    In this case, for any unaffected rich vertex $u \in U$, $vu$ is an arc by Lemma~\ref{LEM: MAX NONTERMINAL INTERVAL}.
    Since $\IL \subseteq \optIL = \IL^{\dag}$, $vu$ is a forward arc w.r.t.\ $\order^{\dag}$.

    \textbf{Case 2:} $v \in \richout$.
    In this case, we have $v \in \optIL \cup \optI = \IL^{\dag} \cup I^{\dag}$.
    Since $|N^{-}_{U'}(u)| = |N^{-}_{U}(u)|$, the number of affected arcs (w.r.t.\ $\optorder$) connecting $v$ and the vertices in $U'$ equals the number of affected arcs (w.r.t.\ $\order^{\dag}$) connecting $v$ and the vertices in $U$.

    \textbf{Case 3:} $v \in \rich \setminus U$.
    In this case, we have $v \in W_{I'} \subseteq \optI \subseteq I^{\dag}$.
    Thus, the arc connecting $v$ and the vertices in $U$ must be unaffected.
    
    \textbf{Case 4:} $v \in \richin$.
    In this case, we have $v \in \optIR \cup \optI = \IR^{\dag} \cup I^{\dag}$.
    Since $|N^{-}_{U'}(u)| = |N^{-}_{U}(u)|$, the number of affected arcs (w.r.t.\ $\optorder$) connecting $v$ and the vertices in $U'$ equals the number affected arcs (w.r.t.\ $\order^{\dag}$) connecting $v$ and the vertices in $U$.

    \textbf{Case 5:} $v \in \IL$.
    In this case, for any unaffected rich vertex $u \in U$, $uv$ is an arc by Lemma~\ref{LEM: MAX NONTERMINAL INTERVAL}.
    Since $\IR \subseteq \optIR = \IL^{\dag}$, $uv$ is a forward arc w.r.t.\ $\order^{\dag}$.

    Overall, the number of affected arcs that disappear after removing $U'$ is no less than the number affected arcs that appear after adding $U$.
    Therefore, we conclude that $\cost{\order^{\dag}} \leq \cost{\optorder}$, implying that $\ins$ is a YES-instance if $\ins'$ is a YES-instance.
    
    Finally, Reduction Rule~\ref{REDUCTION: RICH REPLACE} can be executed in polynomial time, as each of its steps can be executed in time $\bigO{n^{2}} + \bigO{n'^{2}} = \bigO{(n + k)^{2}}$.
    In addition, by Lemma~\ref{LEM: LINEAR LENTH}, after applying Reduction Rule~\ref{REDUCTION: RICH REPLACE}, the length of the maximal interval is less than $(7\appr + 13)k + 1$, meaning that Reduction Rule~\ref{REDUCTION: RICH REPLACE} can be applied at most $n$ times.
    Consequently, Reduction Rule~\ref{REDUCTION: RICH REPLACE} can be applied in polynomial time.
    \qed
\end{proof}

Based on Lemma~\ref{LEM: LINEAR LENTH}, after applying Reduction Rule~\ref{REDUCTION: RICH REPLACE}, the length of each maximal non-terminal is bounded by $\maxlen = (7\appr + 13)k + 5$.
Based on Lemma~\ref{LEM: NONTERMINA INTERVAL LENGTH}, we can derive that the reduced instance with a regular order of vertices is a NO-instance if the number of vertices
\begin{align*}
    n 
    >{}& (2 \appr k + 1) \maxlen + 4(\appr + 1)\appr k^{2} + 4\appr k \\
    ={}& (2 \appr k + 1) ((7\appr + 13)k + 5) + 4(\appr + 1)\appr k^{2} + 4\appr k \\
    ={}& 6\appr (3 \appr + 5)k^{2} + (21\appr + 13) k + 5.
\end{align*}

\begin{reduction} \label{REDUCTION: QUADRATIC}
    Let $\ins$ be a reduced instance to which Reduction Rule~\ref{REDUCTION: RICH REPLACE} cannot be applied.
    If the number of vertices in $\ins$ satisfies that
    \begin{equation*}
        n \geq 6\appr (3 \appr + 5)k^{2} + (21\appr + 13) k + 6,
    \end{equation*}
    then return a trivial NO-instance.
\end{reduction}

Directly following Reduction Rule~\ref{REDUCTION: QUADRATIC}, we obtain a kernel for \SFAST{}.
The size of the kernel depends only on the solution size $k$ and the approximation ratio $\appr$.

\begin{theorem}
    If there is an $\appr$-approximation algorithm for \SFAST{}, then \SFAST{} admits a kernel of vertex size $\bigO{(\appr k)^{2}}$.   
\end{theorem}

It is well-known that \SFAS{} admits an $\bigO{\log k \log \log k}$-approximation algorithm~\cite{combinatoricaSeymour95,algorithmicaEvenNSS98}, which can also be applied to \SFAST{} naturally.
Using this, we obtain an almost quadratic kernel for \SFAST{}.

\begin{corllary}
    \SFAST{} admits a kernel of vertex size $\bigO{(k \log k \log \log k)^{2}} = \tildeO{k^{2}}$.   
\end{corllary}

\section{Concluding Remarks} \label{SEC: CONCLU}

Nearly two decades ago, the \FAST{} problem was shown to admit a linear kernel.
However, \SFAST{} seems more challenging.
In contrast to the \textsc{Subset Feedback Vertex Set in Tournaments} problem, which can be viewed as the implicit \textsc{$3$-Hitting Set} problem and thus admits a quadratic kernel~\cite{tcsBaiX23}, \SFAST{} does not share this property.
Prior to this work, it remained unknown whether \SFAST{} admits a polynomial kernel.

In this paper, by introducing the novel concepts of regular orders and rich vertices, we establish a polynomial kernel for \SFAST{}.
It remains an open question whether this bound can be further improved.
An approximation algorithm with a constant approximation ratio would lead to a quadratic kernel.



\bibliographystyle{splncs04}
\bibliography{SFAST}

\end{document}